\documentclass[11pt]{article}%
\pdfoutput=1
\usepackage{amsfonts}
\usepackage{color}
\usepackage{amsmath}
\usepackage{fullpage}
\usepackage{amssymb}
\usepackage{hyperref}
\usepackage{graphicx}%
\providecommand{\U}[1]{\protect\rule{.1in}{.1in}}
\newtheorem{theorem}{Theorem}

\newtheorem{lemma}[theorem]{Lemma}

\newtheorem{remark}[theorem]{Remark}

\newenvironment{proof}[1][Proof]{\noindent\textbf{#1.} }{\ \rule{0.5em}{0.5em}}

\let\originalleft\left
\let\originalright\right
\renewcommand{\left}{\mathopen{}\mathclose\bgroup\originalleft}
\renewcommand{\right}{\aftergroup\egroup\originalright}

\begin{document}

\title{\textbf{Strong converse for the classical capacity of the pure-loss bosonic
channel}}
\author{
Mark M. Wilde\thanks{Department of Physics and Astronomy, Center for
Computation and Technology, Louisiana State University, Baton Rouge, Louisiana
70803, USA}
\and
Andreas Winter\thanks{ICREA \& F\'{\i}sica Te\`{o}rica: Informaci\'{o} i
Fenomens Qu\`{a}ntics, Universitat Aut\`{o}noma de Barcelona, ES-08193
Bellaterra (Barcelona), Spain} \thanks{School of Mathematics, University of
Bristol, Bristol BS8 1TW, United Kingdom}
}
\date{5 November 2013}

\maketitle

\begin{abstract}
This paper strengthens the interpretation and understanding of the classical
capacity of the pure-loss bosonic channel, first established in [Giovannetti
\textit{et al}., \textit{Physical Review Letters} \textbf{92}, 027902 (2004),
arXiv:quant-ph/0308012]. In particular, we first prove that there exists a
trade-off between communication rate and error probability if one imposes only
a mean-photon number constraint on the channel inputs. That is, if we demand
that the mean number of photons at the channel input cannot be any larger than
some positive number $N_{S}$, then it is possible to respect this constraint
with a code that operates at a rate $g\left(  \eta N_{S}/\left(  1-p\right)
\right)  $ where $p$ is the code's error probability, $\eta$ is the channel transmissivity,
and $g\left(  x\right)  $
is the entropy of a bosonic thermal state with mean photon number~$x$.
We then prove that a strong converse theorem holds for the classical capacity
of this channel (that such a rate-error trade-off cannot occur) if one instead
demands for a maximum photon number constraint, in such a way that mostly all
of the \textquotedblleft shadow\textquotedblright\ of the average density
operator for a given code is required to be on a subspace with photon number
no larger than $nN_{S}$, so that the shadow outside this subspace vanishes as
the number $n$ of channel uses becomes large.
Finally, we prove that 
a small modification of the well-known 
coherent-state coding scheme meets this more demanding constraint.
\end{abstract}

\section*{Introduction}

The pure-loss bosonic channel is one of the most important communication
channels studied in quantum information theory \cite{S09,WPGCRSL12}. It has
acquired this elevated status because it is a simple model for free-space
communication or transmission over fiber optic cables. Indeed, a normal-mode
decomposition of a quantized propagating electromagnetic field in free space
leads naturally to a set of orthogonal spatio-temporal modes \cite{S09}, such
that the $k$th output mode can be expressed in terms of the $k$th input and
environment modes as follows:%
\begin{equation}
\hat{b}_{k}=\sqrt{\eta_{k}}\,\hat{a}_{k}+\sqrt{1-\eta_{k}}\,\hat{e}_{k},
\label{eq:bosonic-channel}%
\end{equation}
where $\hat{a}_{k}$, $\hat{b}_{k}$, and $\hat{e}_{k}$ are the annihilation
operators corresponding to the $k$th field mode of the sender, receiver, and
environment, respectively. 
For the pure-loss channel, the environment modes are originally prepared in 
the vacuum state.
The transmissivity parameter $\eta_{k}\in\left[
0,1\right]$ characterizes (roughly) the fraction of photons that make it
through one of these channels on average to the receiver, and the state
prepared at each environment mode is the vacuum state. When attempting to gain
an information-theoretic understanding of free space communication, it is
often more convenient and simpler to focus on a single-mode channel of the
form in (\ref{eq:bosonic-channel}), rather than the full normal mode
decomposition, and this is what we refer to as the pure-loss bosonic channel.

If we allow for signal states at the input of the pure-loss bosonic channel
that have an arbitrarily large number of photons, then the classical capacity
of this channel is infinite. This is because, with infinite-energy signal
states, one can space them out in such a way that there are an infinite number
of them that are all 
arbitrarily well 
distinguishable from one another, at the channel output. 
Thus, in order to have a sensible notion of classical 
capacity for this channel, we
should impose a constraint on the photon number $a^\dagger a$ 
(or, equivalently, energy $H=a^\dagger a + \frac12$ of the signal
states). One natural constraint is on the mean photon number---i.e., we demand
that the mean number of photons in any codeword transmitted through the
channel should be no larger than some number $N_{S}\geq0$. With such a
constraint, the classical capacity of this channel is equal to $g\left(  \eta
N_{S}\right)  $ \cite{GGLMSY04}, where%
\begin{equation}
g\left(  x\right)  \equiv\left(  x+1\right)  \log_{2}\left(  x+1\right)
-x\log_{2}x \label{eq:thermal-state-entropy}%
\end{equation}
is the entropy of a bosonic thermal state with mean photon number $x$. This result
follows from a proof that there exists a coding scheme that can achieve this
rate \cite{HolevoWerner2001},
and a matching
converse proof that demonstrates it is impossible to have perfectly reliable
communication if the rate exceeds $g\left(  \eta N_{S}\right)  $
\cite{GGLMSY04}.\footnote{The reader might also 
consider related work on classical communication
over noiseless bosonic channels \cite{YO93},
classical communication over pure-state \cite{PhysRevA.54.1869} and  general quantum channels
\cite{PhysRevA.56.131,Hol98}, and other schemes for decoding the pure-loss bosonic channel
\cite{WGTS11}.}

Although Ref.~\cite{GGLMSY04} proved that $g\left(  \eta N_{S}\right)  $ is
equal to the classical capacity of the pure-loss bosonic channel, the converse
theorem used there is only a
\textquotedblleft weak converse,\textquotedblright\ meaning that the upper
bound on the rate $R$ of any coding scheme for this channel with error
probability $\varepsilon$ is of the following form:%
\begin{equation}
R\leq\frac{1}{1-\varepsilon}\left[  g\left(  \eta N_{S}\right)  +h_{2}\left(
\varepsilon\right)  \right]  , \label{eq:weak-converse-bosonic}%
\end{equation}
where $h_{2}\left(  \varepsilon\right)  $ is the binary entropy, with the property that
$\lim_{\varepsilon\rightarrow0}h_{2}\left(  \varepsilon\right)  =0$. Thus, in
order to establish $g\left(  \eta N_{S}\right)  $ as the capacity, one really
needs to take the limit in (\ref{eq:weak-converse-bosonic}) as the error
probability $\varepsilon
\rightarrow0$. This in fact is the hallmark of a weak converse---it leaves
room for a trade-off between rate and error, and suggests that one might be
able to attain a higher communication rate by allowing for some error.

A strong converse theorem demonstrates that there is no such
trade-off in the limit of large blocklength.
That is, a strong converse theorem holds if the error probability
converges to one in the limit of many channel uses when the communication rate
of a coding scheme exceeds the classical capacity. Thus, such a theorem
improves our understanding of the capacity as a sharp dividing line between
what communication rates are possible or impossible, and in this sense, it is
analogous to a phase transition in statistical physics. Furthermore, there are
applications of strong converse theorems in establishing security for
particular models of cryptography \cite{KWW12}.


Several prior works have established the strong converse for the classical
capacity of certain quantum channels. There are two independent proofs 
\cite{W99} and \cite{ON99} that the strong converse theorem holds for any discrete
memoryless channel with a classical input and finite-dimensional quantum
output, so-called ``cq-channels.''
More generally, it holds for arbitrary finite dimensional quantum channels
with product state encoding \cite{ON99,WinterPhD99}.
Many years later, it was shown that the strong converse holds for
covariant channels for which their maximum output $p$-norm is multiplicative
\cite{KW09}. Finally, recent work has shown that the strong converse holds
for all entanglement-breaking and Hadamard channels \cite{WWY13}.

\section*{Summary of results}

This paper establishes several facts regarding classical communication over
the pure-loss bosonic channel:

\begin{enumerate}
\item If we demand only that the mean photon number is no larger than some
number $N_{S}\geq0$, then we show that the strong converse does not hold (see
Section~\ref{sec:no-s-c-mean-ph}). We can show that this is the case even if
we restrict the codewords to be pure states. In some sense, this latter result
provides a distinction between the classical and quantum theories of
information for continuous variables, but it does bear some similarities with
the observations in Theorem~77 of Ref.~\cite{P10} and we remark on this point
further in Section~\ref{sec:no-s-c-mean-ph}.

\item In light of the above result, we can only hope to prove that the strong
converse holds under some alternate photon number constraint. Let $\rho_{m}$
denote an $n$-mode codeword in a given codebook, so that the average code
density operator is given by $\frac{1}{M}\sum_{m}\rho_{m}$, where
$M$ is the total number of messages. We instead demand
that the average code density operator satisfies a maximum photon number
constraint, such that it should have a large
``shadow'' onto a subspace of photon number no more than $\left\lceil
nN_{S}\right\rceil $:%
\begin{equation}
\frac{1}{M}\sum_{m}\text{Tr}\left\{  \Pi_{\left\lceil nN_{S}\right\rceil }%
\rho_{m}\right\}  \geq1-\delta\left(  n\right)  .
\label{eq:photon-number-constraint-1}%
\end{equation}
In the above, $\Pi_{\left\lceil nN_{S}\right\rceil }$ is the projector onto a subspace
with photon number no larger than $\left\lceil nN_{S}\right\rceil $ and
$\delta\left(  n\right)  $ decreases to zero with increasing $n$. Under such a
constraint, we prove in Section~\ref{sec:s-c-max-p-n-constraint} that the
strong converse holds for the classical capacity of the pure-loss bosonic channel.

\item Finally, we prove in Section~\ref{sec:code-existence}\ that there exist
codes for the pure-loss bosonic channel that meet the constraint in
(\ref{eq:photon-number-constraint-1}) while having an error probability that
can be less than an arbitrarily small constant for a sufficiently large number
of channel uses. Indeed, we show that the usual coherent state encoding 
scheme essentially satisfies the constraint. 
\end{enumerate}

\section{No strong converse under a mean photon number constraint}

\label{sec:no-s-c-mean-ph}We first prove that a strong converse does not hold
for the classical capacity of the pure-loss bosonic channel if we impose only
a mean photon number constraint. Indeed, a method for proving the existence of
a code that achieves the classical capacity of the pure-loss bosonic channel
is first to sample coherent-state codewords independently from a
circularly-symmetric complex Gaussian distribution with variance $N_{S}$
\cite{HolevoWerner2001,GGLMSY04}. Let%
\[
\left\vert \alpha^{n}\left(  m\right)  \right\rangle \equiv\left\vert
\alpha_{1}\left(  m\right)  \right\rangle \otimes\cdots\otimes\left\vert
\alpha_{n}\left(  m\right)  \right\rangle
\]
denote each of the $n$-mode coherent state codewords (with the dependence on
the message $m$ explicitly indicated)\ and let $\left[  M\right]  $ denote a
message set of size $M$. Then one can prove that there exists a choice of
codebook such that every codeword in the resulting codebook $\left\{
\left\vert \alpha^{n}\left(  m\right)  \right\rangle \right\}  _{m\in\left[
M\right]  }$ satisfies the mean photon number constraint. Let $\left\vert
\beta^{n}\left(  m\right)  \right\rangle $ denote the state resulting from
sending the codeword $\left\vert \alpha^{n}\left(  m\right)  \right\rangle $
through the pure-loss bosonic channel, so that%
\begin{equation}
\left\vert \beta^{n}\left(  m\right)  \right\rangle \equiv\left\vert
\sqrt{\eta}\alpha_{1}\left(  m\right)  \right\rangle \otimes\cdots
\otimes\left\vert \sqrt{\eta}\alpha_{n}\left(  m\right)  \right\rangle
,\label{eq:output-states}%
\end{equation}
and $\eta$ is the channel transmissivity parameter. Furthermore, this choice
of codebook is such that the receiver can decode the transmitted codewords
with arbitrarily high success probability by performing a square-root
measurement \cite{PhysRevA.54.1869}\ or a sequential decoding measurement
\cite{WGTS11}, for example. That is, as long as $\left(  \log_{2}M\right)
/n\approx g\left(  \eta N_{S}\right)  $ and the number $n$\ of channel uses is
sufficiently large, there exists a measurement $\left\{  \Lambda_{m}\right\}
_{m\in [M]}$ and codebook such that%
\[
\forall m\in\left[  M\right]  :\text{Tr}\left\{  \Lambda_{m}\left\vert
\beta^{n}\left(  m\right)  \right\rangle \left\langle \beta^{n}\left(
m\right)  \right\vert \right\}  \geq1-\varepsilon,
\]
where $\varepsilon$ is an arbitrarily small positive number. For finite $n$,
$M$, and $\varepsilon$, we call such a code an $(n,M,\varepsilon)$
code. 

\subsection{No strong converse with mixed-state codewords}

Now, consider a pure-loss bosonic channel with transmissivity parameter $\eta$
and mean photon number constraint $N_{S}$. As we have stated before, the
classical capacity of this channel is equal to $g\left(  \eta N_{S}\right)  $ \cite{GGLMSY04}.
We show that a strong converse theorem cannot hold---our approach is to
employ a codebook of the above form at a rate larger than $g\left(  \eta
N_{S}\right)  $. Indeed, let $\left\{  \left\vert \alpha^{n}\left(  m\right)
\right\rangle \right\}  _{m\in\left[  M\right]  }$ now denote a codebook such
that each codeword has mean photon number $P>N_{S}$, let the rate of the code
be $\left(  \log_{2}M\right)  /n\approx g\left(  \eta P\right)  >g\left(  \eta
N_{S}\right)  $, and let $\left\{  \Lambda_{m}\right\}  _{m\in\left[
M\right]  }$ be a decoding measurement for this codebook. We then modify the
codewords as follows:%
\begin{equation}
\rho\left(  m\right)  \equiv\left(  1-p\right)  \left\vert \alpha^{n}\left(
m\right)  \right\rangle \left\langle \alpha^{n}\left(  m\right)  \right\vert
+p\left(  \left\vert 0\right\rangle \left\langle 0\right\vert \right)
^{\otimes n},\label{eq:mixture-codeword}%
\end{equation}
where $p$ is such that$\ 0\leq p\leq1$ and $\left(  1-p\right)  P=N_{S}$ and
$\left\vert 0\right\rangle ^{\otimes n}$ is the $n$-fold tensor product vacuum
state. Observe that the mean photon number of each codeword $\rho\left(
m\right)  $ is equal to $\left(  1-p\right)  P=N_{S}$ (because the vacuum
contributes nothing to the photon number). Then the state resulting from
transmitting these codewords through the pure-loss bosonic channel is as
follows:%
\begin{equation}
\left(  1-p\right)  \left\vert \beta^{n}\left(  m\right)  \right\rangle
\left\langle \beta^{n}\left(  m\right)  \right\vert +p\left(  \left\vert
0\right\rangle \left\langle 0\right\vert \right)  ^{\otimes n}%
,\label{eq:output-mixture}%
\end{equation}
with $\left\vert \beta^{n}\left(  m\right)  \right\rangle $ defined in the
same way as in (\ref{eq:output-states}). The success probability for every
codeword is as follows:%
\begin{align}
  \text{Tr}\left\{  \Lambda_{m}\left(  \left(  1-p\right)  \left\vert
\beta^{n}\left(  m\right)  \right\rangle \left\langle \beta^{n}\left(
m\right)  \right\vert +p\left(  \left\vert 0\right\rangle \left\langle
0\right\vert \right)  ^{\otimes n}\right)  \right\}  
&  \geq\left(  1-p\right)  \text{Tr}\left\{  \Lambda_{m}\left\vert \beta
^{n}\left(  m\right)  \right\rangle \left\langle \beta^{n}\left(  m\right)
\right\vert \right\}  \label{eq:decoder-perf-1}\\
&  \geq\left(  1-p\right)  \left(  1-\varepsilon\right)
,\label{eq:decoder-perf-2}%
\end{align}
simply by using the fact that the decoding measurement $\left\{  \Lambda
_{m}\right\}  $ decodes the codewords $\left\{  \left\vert \beta^{n}\left(
m\right)  \right\rangle \right\}  $ with success probability larger than
$1-\varepsilon$. Thus, the success probability need not converge to zero in
the limit of many channel uses if the rate of the code exceeds the classical
capacity and we impose only a mean photon number constraint on each codeword
in the codebook.

\begin{remark}
\label{rem:q-cl-similarity}There is nothing particularly \textquotedblleft
quantum\textquotedblright\ about the above argument, other than the fact that
we employ mixtures of quantum states as codewords and a quantum measurement
for decoding. Indeed, the same argument demonstrates that a strong converse
does not hold when we impose only a mean power constraint for each classical
codeword transmitted over the classical additive white Gaussian noise channel
but allow for probabilistic encodings. That is, for a given codebook where
each codeword has mean power $P$, we can always produce a new codebook such
that each codeword is a $\left(  1-p,p\right)  $\ Bernoulli mixture of a
codeword from the original codebook and the all-zero signal. One then produces
a codebook where each codeword has mean power $\left(  1-p\right)  P$, and we
can violate the strong converse using an argument very similar to the above
one. However, if we restrict to deterministic or classical \textquotedblleft
pure-state\textquotedblright\ encodings, then a strong converse theorem does
hold as shown in Ref.~\cite{P10}.
\end{remark}

\subsection{No strong converse with pure-state codewords}

This section demonstrates an important distinction between the classical and
quantum theories of information for continuous variables: we show that a
strong converse does not hold when imposing a mean photon-number constraint
and even when we restrict to pure-state codewords (see
Remark~\ref{rem:q-cl-similarity}\ above). Our argument is similar to the one
in the previous section.

We are again considering the pure-loss bosonic channel with transmissivity
parameter $\eta$ and mean photon number constraint $N_{S}$. We will show the
existence of a codebook with pure-state codewords each with mean photon number
$N_{S}$ such that the rate is strictly larger than $g\left(  \eta
N_{S}\right)  $ while the success probability of decoding is bounded from
below by a constant number between zero and one.

The idea is similar to that in the previous section, except we make our
codewords be superpositions of the following form:%
\begin{equation}
\left\vert \gamma_{p}\left(  m\right)  \right\rangle \equiv\sqrt
{1-p}\left\vert \alpha^{n}\left(  m\right)  \right\rangle \left\vert
0\right\rangle +\sqrt{p}\left\vert 0\right\rangle ^{\otimes n}\left\vert
1\right\rangle , \label{eq:gamma-states}
\end{equation}
so that we add just one more mode to send through the channel that has a
negligible effect on the parameters of the code (an $(n,M,\varepsilon)$ code
becomes an $(n+1, M, \varepsilon)$ code, which is a negligible change for
large $n$). This extra mode uses a single
photon to purify the state in (\ref{eq:mixture-codeword}). The mean photon
number of the codeword $\left\vert \gamma_{p}\left(  m\right)  \right\rangle $
is equal to%
\begin{align*}
&  \text{Tr}\left\{  \frac{1}{n+1}\sum_{i=1}^{n+1}\hat{a}_{i}^{\dag}\hat
{a}_{i}\left\vert \gamma_{p}\left(  m\right)  \right\rangle \left\langle
\gamma_{p}\left(  m\right)  \right\vert \right\} \\
&  =\text{Tr}\left\{  \frac{1}{n+1}\sum_{i=1}^{n+1}\hat{a}_{i}^{\dag}\hat
{a}_{i}\left[
\begin{array}
[c]{c}%
\left(  1-p\right)  \left\vert \alpha^{n}\left(  m\right)  \right\rangle
\left\langle \alpha^{n}\left(  m\right)  \right\vert \otimes\left\vert
0\right\rangle \left\langle 0\right\vert +\sqrt{\left(  1-p\right)
p}\left\vert \alpha^{n}\left(  m\right)  \right\rangle \left\langle
0\right\vert ^{\otimes n}\otimes\left\vert 0\right\rangle \left\langle
1\right\vert \\
+\sqrt{\left(  1-p\right)  p}\left\vert 0\right\rangle ^{\otimes
n}\left\langle \alpha^{n}\left(  m\right)  \right\vert \otimes\left\vert
1\right\rangle \left\langle 0\right\vert +p\left(  \left\vert 0\right\rangle
\left\langle 0\right\vert \right)  ^{\otimes n}\otimes\left\vert
1\right\rangle \left\langle 1\right\vert
\end{array}
\right]  \right\} \\
&  =\left(  1-p\right)  \frac{nP}{n+1}+p\frac{1}{n+1}.
\end{align*}
Thus, again by tuning $p$ and choosing $n$ large enough, we can set $\left(
1-p\right)  \frac{nP}{n+1}+p\frac{1}{n+1}=N_{S}$. The receiver operates by
tracing over the very last mode, which dephases the outputs to be of the form
in (\ref{eq:output-mixture}), and he then uses the decoding measurement for
the codebook $\left\{  \left\vert \alpha^{n}\left(  m\right)  \right\rangle
\right\}  $. More explicitly, the receiver measures the $n+1$ output modes
with the decoding POVM $\{\Lambda_m \otimes I\}$,
where the operators $\Lambda_m$ are the same as those from the previous section
(acting on the first $n$ modes) and
the identity operator acts on the very last mode. We then find the following bound on
the success probability for every codeword:
\begin{align*}
\operatorname{Tr}\{ (\Lambda_m \otimes I) \mathcal{N}^{\otimes (n+1)}
(\left\vert \gamma_{p}\left(  m\right)  \right\rangle \left\langle \gamma_{p}\left(  m\right)  \right\vert )\}
& = \operatorname{Tr}\{ \Lambda_m  \mathcal{N}^{\otimes n}
(\rho(m) )\} \\
& = \text{Tr}\left\{  \Lambda_{m}\left(  \left(  1-p\right)  \left\vert
\beta^{n}\left(  m\right)  \right\rangle \left\langle \beta^{n}\left(
m\right)  \right\vert +p\left(  \left\vert 0\right\rangle \left\langle
0\right\vert \right)  ^{\otimes n}\right)  \right\}  \\
& \geq (1-p)(1-\varepsilon).
\end{align*}
The first equality follows from the fact that the channel is
completely positive and trace preserving
and from the fact that the state $\left\vert \gamma_{p}\left(  m\right)  \right\rangle$
defined in (\ref{eq:gamma-states})
is a purification of the state $\rho(m)$ defined in (\ref{eq:mixture-codeword}).
The next equality is from the definition of $\rho(m)$.
The last inequality follows from the argument in~(\ref{eq:decoder-perf-1}%
)-(\ref{eq:decoder-perf-2}).
The resulting code operates at a rate $\approx g\left(  \eta P\right)  $ while meeting
the mean photon number constraint. Thus, a strong converse theorem cannot hold
for the pure-loss bosonic channel with only a mean photon number constraint
and when restricting to pure-state codewords.

\begin{remark}
The coding scheme given both in this section and the previous one demonstrate that
we can achieve a rate-error trade-off of $( g(\eta N_S / (1-p)) , p)$ for all
$0 \leq p \leq 1$, where $p$ is the error probability and $N_S$ is the mean
photon number constraint. This result complements the classical result in Theorem 77 of
Ref.~\cite{P10}. However, in our case,
we have not proved that this trade-off is optimal, merely that
it is achievable.
\end{remark}

\section{Strong converse under a maximum photon number constraint}

\label{sec:s-c-max-p-n-constraint}
In this section, we prove that the strong converse holds when imposing 
a particular maximum photon number constraint.
Our approach for proving the strong converse
theorem is a simulation based argument, similar to that used in
Refs.~\cite{W02,BDHSW09,BBCW12,BBCW13}. We can illustrate the main idea behind
this argument by recalling a simple proof of the strong converse theorem for
the noiseless qubit channel \cite{KW09,N99}. Consider that any scheme for
classical communication over $n$ noiseless qubit channels consists of an
encoding of a message $m$\ as a quantum state $\rho_{m}$\ on $n$ qubits,
followed by a decoding measurement $\left\{  \Lambda_{m}\right\}  $ to recover
the message. Let $M$ be the total number of messages, so that the rate of the
code is $R\equiv\frac{1}{n}\log_{2}\left(  M\right)  $. The average success
probability of this scheme is bounded as follows:%
\begin{align}
\frac{1}{M}\sum_{m}\text{Tr}\left\{  \Lambda_{m}\rho_{m}\right\}   &
\leq\frac{1}{M}\sum_{m}\text{Tr}\left\{  \Lambda_{m}\right\} \nonumber\\
&  =M^{-1}2^{n}\nonumber\\
&  =2^{-n\left(  R-1\right)  }, \label{eq:trivial-SC}%
\end{align}
where the first inequality follows from the operator inequality $\rho_{m}\leq
I$, which holds for any density operator, and the first equality follows because
$\sum_{m}\Lambda_{m}=I^{\otimes n}$ (the identity operator on $n$ qubits).
Thus, it is clear that if the rate $R$ exceeds one, then the average success
probability of any communication scheme decreases exponentially fast to zero
with increasing blocklength.

Our proof of the strong converse for the classical capacity of the pure-loss
bosonic channel is similar in spirit to the above argument, but it requires
some nontrivial additions. First, we show that there is a simple protocol
using approximately $ng\left(  \eta N_{S}\right)  $ noiseless qubit channels
to faithfully simulate the action of $n$ instances of the pure-loss bosonic
channel with transmissivity $\eta$ when acting on a space with total photon
number less than $nN_{S}$, such that the simulation error becomes arbitrarily
small as $n$ becomes large. So we can combine this simulation protocol with
any classical code $\mathcal{C}$ for the pure-loss bosonic channel in which
each codeword has almost all of its ``shadow'' on a space with total photon number
less than $nN_{S}$. Letting $\mathcal{N}$ denote the pure-loss bosonic
channel, we can phrase this simulation argument in the language of resource
inequalities \cite{DHW05RI}\ as follows:%
\[
g\left(  \eta N_{S}\right)  \left[  q\rightarrow q\right]  \geq\left\langle
\mathcal{N}:\mathcal{C}\right\rangle .
\]
The meaning of the above resource inequality is
that one can simulate the action of $n$ instances of the pure-loss bosonic channel
on the codewords in$~\mathcal{C}$ by exploiting noiseless qubit channels at a
rate equal to $g\left(  \eta N_{S}\right)  $, and this simulation becomes perfect
in the limit as $n$ becomes large. If one could send classical
information over the pure-loss bosonic channel at a rate $R$\ larger than
$g\left(  \eta N_{S}\right)  $, then it would be possible to serially
concatenate the above protocol with a classical coding scheme for the channel
$\mathcal{N}$ and achieve the following resource inequality:%
\[
g\left(  \eta N_{S}\right)  \left[  q\rightarrow q\right]  \geq\left\langle
\mathcal{N}:\mathcal{C}\right\rangle \geq R\left[  c\rightarrow c\right]  ,
\]
where $R\left[  c\rightarrow c\right]$ denotes noiseless classical communication
at rate $R$.
Since the above protocol would give a strong violation of the Holevo bound for
$R>g\left(  \eta N_{S}\right)  $ (in particular, the refinement given in
(\ref{eq:trivial-SC})), it must not be possible. In fact, essentially the same
argument as in (\ref{eq:trivial-SC}) demonstrates that the error probability
of the classical communication protocol goes exponentially fast to one if the
rate $R$ is strictly larger than the classical capacity of this channel.
Figure~\ref{fig:simulation-argument}\ depicts this simulation argument.%
\begin{figure}[ht]
\begin{center}
\includegraphics[
natheight=3.594200in,
natwidth=8.666300in,
height=2.3921in,
width=5.7285in
]%
{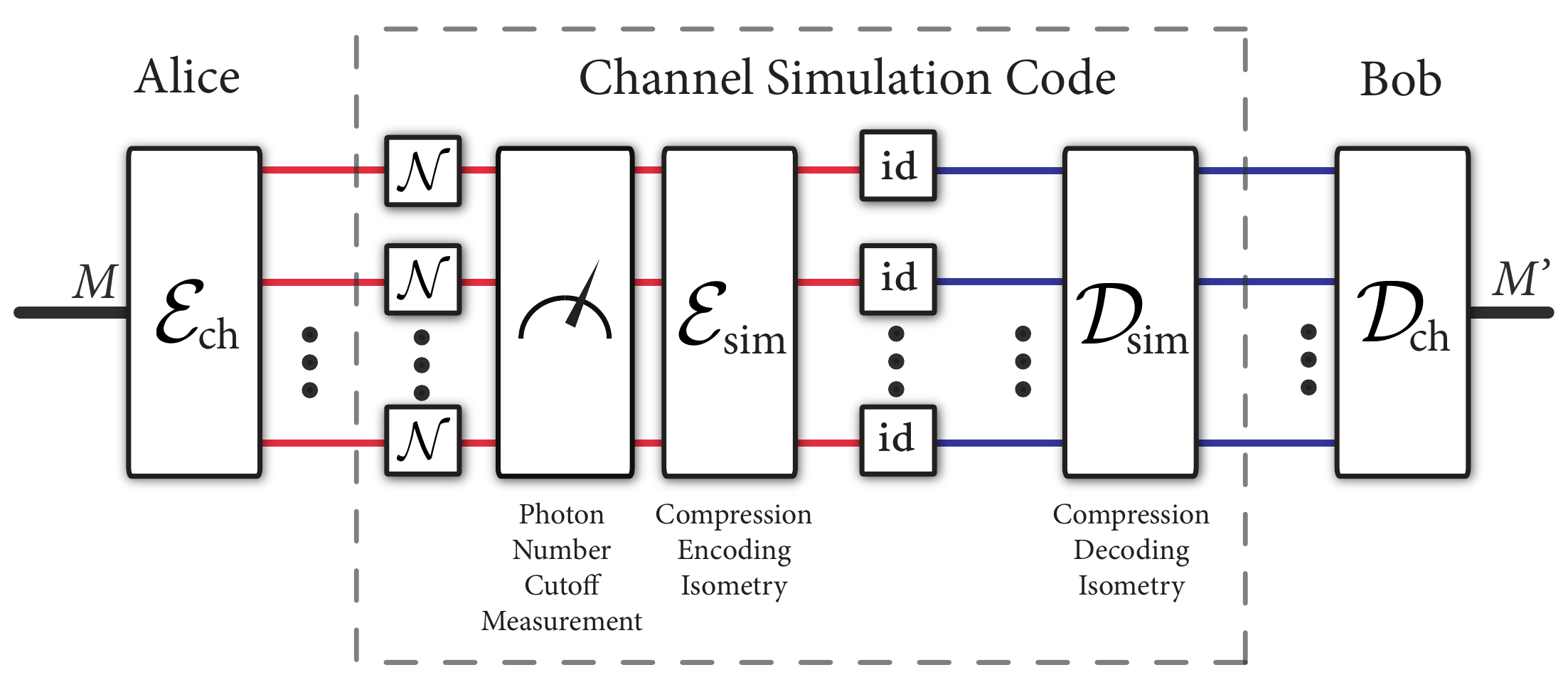}%
\caption{The simulation argument for the strong converse of the classical
capacity of the pure-loss bosonic channel with transmissivity $\eta\in\left[
0,1\right]  $. If an $n$-mode quantum state input to $n$ uses of the pure-loss
bosonic channel has nearly all of its ``shadow'' on a subspace with photon number no
larger than $nN_{S}$, then one can simulate this channel with high fidelity on
any such input at a rate of $g\left(  \eta N_{S}\right)  $ qubits per mode. If
it were possible to send classical information over the pure-loss bosonic
channel at a rate $R>g\left(  \eta N_{S}\right)  $, then Alice and Bob could
combine a channel code with the simulation code and violate the Holevo
bound. Analyzing this contradiction in more detail allows us to conclude the
strong converse theorem. The simulation protocol in the figure begins with
Alice encoding a message $M$ using the channel encoder $\mathcal{E}%
_{\text{ch}}$. Alice proceeds to simulate the channel by first actually
performing it on all of the input modes and then projecting onto a subspace
with photon number no larger than $\approx n\eta N_{S}$. This measurement
succeeds with high probability and then she can compress the quantum data at
the output of the measurement at a rate of $g\left(  \eta N_{S}\right)  $
qubits per mode (sending this quantum data over noiseless qubit channels, denoted by ``id''
in the figure for ``identity'' channels). Bob then decompresses the quantum data, completing the
channel simulation, and he finally decodes the classical message using the
channel decoder $\mathcal{D}_{\text{ch}}$.}%
\label{fig:simulation-argument}%
\end{center}
\end{figure}


Let $\Pi_{L}$ denote the projector onto a subspace of $n$ bosonic modes, such
that the total photon number is no larger than $L$:%
\[
\Pi_{L}\equiv\sum_{a_{1},\ldots,a_{n}:\sum_{i}a_{i}\leq L}\left\vert
a_{1}\right\rangle \left\langle a_{1}\right\vert \otimes\cdots\otimes
\left\vert a_{n}\right\rangle \left\langle a_{n}\right\vert ,
\]
where $\left\vert a_{i}\right\rangle $ is a photon number state of photon
number $a_{i}$. We call $\Pi_{L}$ the \textquotedblleft photon number cutoff
projector\textquotedblright\ in what follows.

\bigskip

We begin by proving two important lemmas.

\begin{lemma}
\label{lem:energy-entropy-dimension}
The rank of the photon number cutoff
projector $\Pi_{\left\lceil nN_{S}\right\rceil }$, where
$\lceil \cdot \rceil$ denotes the ceiling function, is no larger than%
\[
  2^{n\left[  g\left(  N_{S}\right)  +\delta\right]  },
\]
where
\[
  \delta\geq\frac{1}{n}\left( \log e + \log\left( 1 + \frac{1}{N_{S}} \right) \right).
\]
(Clearly, we can pick $\delta$ to be an arbitrarily small positive constant by
taking $n$ to be large enough.)
\end{lemma}

\begin{proof}
Consider that the rank of the photon number cutoff projector is exactly equal
to%
\[
  \sum_{j=0}^{\left\lceil nN_{S}\right\rceil }\binom{j+n-1}{n-1}
                               = \binom{\left\lceil nN_{S}\right\rceil + n}{n}.
\]
This is because we can enumerate the photon number eigenstates as tuples
$\left(  a_{1},\ldots,a_{n}\right)$ of non-negative integers such that
$\sum_{i}a_{i}\leq\left\lceil nN_{S}\right\rceil$,
which equals the number of tuples $\left( a_0, a_{1},\ldots,a_{n}\right)$
of $a_i \geq 0$ such that $\sum_{i=0}^n a_{i} = \left\lceil nN_{S}\right\rceil$.
In other words, we wish to count the unordered partitions of 
$\left\lceil nN_{S}\right\rceil$ into $n+1$ non-negative integer pieces, 
and these are in turn in one-to-one correspondence with selecting 
$n$ \textquotedblleft separator\textquotedblright\ positions in 
$\left\lceil nN_{S}\right\rceil + n$ to break a block of 
$\left\lceil nN_{S}\right\rceil+n$ into $n+1$ pieces of positive integer 
size, which is given by the binomial coefficient on the right hand side
above. 
We can then bound the rank of $\Pi_{\left\lceil nN_{S}\right\rceil }$ 
as follows:
\begin{align*}
  \operatorname{rank} \Pi_{\left\lceil nN_{S}\right\rceil }
     =    \sum_{j=0}^{\left\lceil nN_{S}\right\rceil }\binom{j+n-1}{n-1}
    &=    \binom{\left\lceil nN_{S}\right\rceil + n}{n}                       \\
    &\leq 2^{\left(  \left\lceil nN_{S}\right\rceil +n\right)
             h_{2}\bigl(  n/\left(  \left\lceil nN_{S}\right\rceil +n\right) \bigr)  }\\
    &\leq 2^{n\left[  g\left(  N_{S}\right)  + \delta\right]  }%
\end{align*}
where $\delta\geq\frac{1}{n}\left( \log e + \log\left( 1 + \frac{1}{N_{S}} \right) \right)$.
The first line of identities we have just argued, and the
inequality in the second line is a well known
combinatorial result (see Example~11.1.3 of Ref.~\cite{CT06}), where
$h_{2}\left(  p\right)  \equiv-p\log_{2}p-\left(  1-p\right)  \log_{2}\left(
1-p\right)  $. 
The last inequality follows from%
\begin{align*}
&\!\!\!\!\!\!\!\!\!\!\!\!\!\!\!
   \left(  \left\lceil nN_{S}\right\rceil +n\right) \,\, h_{2}\left(  n/\left(
\left\lceil nN_{S}\right\rceil +n\right)  \right) \\
&  =\left(  \left\lceil nN_{S}\right\rceil +n\right)  \left[  -\frac
{n}{\left\lceil nN_{S}\right\rceil +n}\log\left(  \frac{n}{\left\lceil
nN_{S}\right\rceil +n}\right)  -\frac{\left\lceil nN_{S}\right\rceil
}{\left\lceil nN_{S}\right\rceil +n}\log\left(  \frac{\left\lceil
nN_{S}\right\rceil }{\left\lceil nN_{S}\right\rceil +n}\right)  \right] \\
&  =-n\log\left(  \frac{n}{\left\lceil nN_{S}\right\rceil +n}\right)
-\left\lceil nN_{S}\right\rceil \log\left(  \frac{\left\lceil nN_{S}%
\right\rceil }{\left\lceil nN_{S}\right\rceil +n}\right) \\
&  =n\log\left(  \frac{\left\lceil nN_{S}\right\rceil +n}{n}\right)
+\left\lceil nN_{S}\right\rceil \log\left(  1+\frac{n}{\left\lceil
nN_{S}\right\rceil }\right) \\
&  \leq n\log\left(  \frac{nN_{S}+n+1}{n}\right)  +\left(  nN_{S}+1\right)
\log\left(  1+\frac{1}{N_{S}}\right) \\
&  \leq n\log\left(  N_{S}+1\right)  + \log e +nN_{S}\log\left(  1+\frac{1}{N_{S}%
}\right)  +\log\left(  1+\frac{1}{N_{S}}\right) \\
&  =ng\left(  N_{S}\right)  + \log e +\log\left(  1+\frac{1}{N_{S}}\right)  .
\end{align*}
The first few equalities are simple algebra. The first inequality follows from
the definition of the ceiling function. The second inequality follows because
$\ln\left(  x+y\right)  \leq\ln\left(  x\right)  +y$ for $x\geq1$ and
$y\geq0$. The last equality uses the definition of $g\left(  x\right)  $ in
(\ref{eq:thermal-state-entropy}).
\end{proof}



\begin{lemma}
\label{lem:input-output-size}Let $\rho^{\left(  n\right)  }$ be a density
operator on $n$ modes such that%
\[
\operatorname{Tr}\left\{  \Pi_{\left\lceil nN_{S}\right\rceil }\rho^{\left(
n\right)  }\right\}  \geq1-\delta_{1},
\]
for some small $\delta_{1}>0$. Then%
\[
\operatorname{Tr}\left\{  \Pi_{\left\lceil n\left(  \eta N_{S}+\delta
_{2}\right)  \right\rceil }\mathcal{N}^{\otimes n}(\rho^{\left(  n\right)
})\right\}  \geq1-2\sqrt{\delta_{1}}-\delta_{1}-\exp\left\{  -\left(
2\delta_{3}^{2}\eta N_{S}\right)  n\right\}  ,
\]
where $\mathcal{N}^{\otimes n}$ represents $n$ instances of the pure-loss
bosonic channel with transmissivity $\eta$, and $\delta_{2}$ and $\delta_{3}$
are fixed positive constants such that $0<\delta_{3}\leq\frac{1}{\left\lceil
nN_{S}\right\rceil }\left(  n\delta_{2}-\eta\right)  $ for $n$ large enough.
\end{lemma}

\begin{proof}
We begin by observing that%
\begin{align}
&\!\!\!\!\!\!\!\!
   \text{Tr}\left\{  \Pi_{\left\lceil n\left(  \eta N_{S}+\delta_{2}\right)
\right\rceil }\mathcal{N}^{\otimes n}(\rho^{\left(  n\right)  })\right\}
\nonumber\\
&  \geq\text{Tr}\left\{  \Pi_{\left\lceil n\left(  \eta N_{S}+\delta
_{2}\right)  \right\rceil }\mathcal{N}^{\otimes n}\left(  \Pi_{\left\lceil
nN_{S}\right\rceil }\rho^{\left(  n\right)  }\Pi_{\left\lceil nN_{S}%
\right\rceil }\right)  \right\}  -\left\Vert \mathcal{N}^{\otimes n}%
(\rho^{\left(  n\right)  })-\mathcal{N}^{\otimes n}\left(  \Pi_{\left\lceil
nN_{S}\right\rceil }\rho^{\left(  n\right)  }\Pi_{\left\lceil nN_{S}%
\right\rceil }\right)  \right\Vert _{1}\nonumber\\
&  \geq\text{Tr}\left\{  \Pi_{\left\lceil n\left(  \eta N_{S}+\delta
_{2}\right)  \right\rceil }\mathcal{N}^{\otimes n}\left(  \Pi_{\left\lceil
nN_{S}\right\rceil }\rho^{\left(  n\right)  }\Pi_{\left\lceil nN_{S}%
\right\rceil }\right)  \right\}  -\left\Vert \rho^{\left(  n\right)  }%
-\Pi_{\left\lceil nN_{S}\right\rceil }\rho^{\left(  n\right)  }\Pi
_{\left\lceil nN_{S}\right\rceil }\right\Vert _{1}\nonumber\\
&  \geq\text{Tr}\left\{  \Pi_{\left\lceil n\left(  \eta N_{S}+\delta
_{2}\right)  \right\rceil }\mathcal{N}^{\otimes n}\left(  \Pi_{\left\lceil
nN_{S}\right\rceil }\rho^{\left(  n\right)  }\Pi_{\left\lceil nN_{S}%
\right\rceil }\right)  \right\}  -2\sqrt{\delta_{1}}.
\label{eq:first-block-sim-proof}%
\end{align}
The first inequality is a special case of the inequality%
\begin{equation}
\text{Tr}\left\{  \Lambda\rho\right\}  \geq\text{Tr}\left\{  \Lambda
\sigma\right\}  -\left\Vert \rho-\sigma\right\Vert _{1},
\label{eq:trace-inequality}%
\end{equation}
which holds for $\Lambda$, $\rho$, $\sigma$ such that $0\leq\Lambda\leq I$,
$0\leq\rho,\sigma$, and Tr$\left\{  \rho\right\}  $, Tr$\left\{
\sigma\right\}  \leq1$. The second inequality follows from the monotonicity of
trace distance under quantum operations. The third inequality is a consequence
of the Gentle Operator Lemma~\cite{W99,ON07}, which states that $\Vert
\rho-\sqrt{\Lambda}\rho\sqrt{\Lambda}\Vert _{1}\leq2\sqrt{\varepsilon}$
if Tr$\left\{  \Lambda\rho\right\}  \geq1-\varepsilon$ for $0\leq
\varepsilon\leq1$, $0\leq\Lambda\leq I$, $\rho\geq0$, Tr$\left\{
\rho\right\}  \leq1$. We can then expand $\Pi_{\left\lceil nN_{S}\right\rceil
}\rho^{\left(  n\right)  }\Pi_{\left\lceil nN_{S}\right\rceil }$ as follows in
the number-state basis:%
\[
\sum_{\substack{a^{n},b^{n}:\\\sum_{i}a_{i},\sum_{i}b_{i}\leq\left\lceil
nN_{S}\right\rceil }}\rho_{a^{n},b^{n}}\left\vert a^{n}\right\rangle
\left\langle b^{n}\right\vert ,
\]
where $a^{n}\equiv a_{1}\cdots a_{n}$, $b^{n}\equiv b_{1}\cdots b_{n}$, and
$\rho^{\left(  n\right)  }=\sum_{a^{n},b^{n}}\rho_{a^{n},b^{n}}\left\vert
a^{n}\right\rangle \left\langle b^{n}\right\vert $. Observe that the
hypothesis of the theorem is equivalent to%
\[
\sum_{a^{n}:\sum_{i}a_{i}\leq\left\lceil nN_{S}\right\rceil }\rho_{a^{n}%
,a^{n}}\geq1-\delta_{1}.
\]
The remaining term in (\ref{eq:first-block-sim-proof}) is equal to%
\begin{multline}
\text{Tr}\left\{  \Pi_{\left\lceil n\left(  \eta N_{S}+\delta_{2}\right)
\right\rceil }\mathcal{N}^{\otimes n}\left(  \Pi_{\left\lceil nN_{S}%
\right\rceil }\rho^{\left(  n\right)  }\Pi_{\left\lceil nN_{S}\right\rceil
}\right)  \right\} \label{eq:term-to-sub-into}\\
=\text{Tr}\left\{  \left(  \Pi_{\left\lceil n\left(  \eta N_{S}+\delta
_{2}\right)  \right\rceil }\otimes I\right)  U^{\otimes n}\left(
\Pi_{\left\lceil nN_{S}\right\rceil }\rho^{\left(  n\right)  }\Pi_{\left\lceil
nN_{S}\right\rceil }\otimes\left\vert 0\right\rangle \left\langle 0\right\vert
^{\otimes n}\right)  \left(  U^{\dag}\right)  ^{\otimes n}\right\} \\
=\sum_{\substack{a^{n},b^{n}:\\\sum_{i}a_{i},\sum_{i}b_{i}\leq\left\lceil
nN_{S}\right\rceil }}\rho_{a^{n},b^{n}}\text{Tr}\left\{  \left(
\Pi_{\left\lceil n\left(  \eta N_{S}+\delta_{2}\right)  \right\rceil }\otimes
I\right)  U^{\otimes n}\left(  \left\vert a^{n}\right\rangle \left\langle
b^{n}\right\vert \otimes\left\vert 0\right\rangle \left\langle 0\right\vert
^{\otimes n}\right)  \left(  U^{\dag}\right)  ^{\otimes n}\right\}  ,
\end{multline}
where $U$ is the unitary transformation corresponding to a beamsplitter with
transmissivity $\eta$. Now, we recall briefly how a beamsplitter of
transmissivity $\eta$ acts on the joint state $\left\vert n\right\rangle
_{A}\left\vert 0\right\rangle _{E}$:%
\begin{align*}
\left\vert n\right\rangle _{A}\left\vert 0\right\rangle _{E}  &
=\frac{\left(  \hat{a}^{\dag}\right)  ^{n}}{\sqrt{n!}}\left\vert
0\right\rangle _{A}\left\vert 0\right\rangle _{E}\\
&  \mapsto \frac{\left(  \sqrt{\eta}\hat{a}^{\dag}+\sqrt{1-\eta}\hat
{e}^{\dag}\right)  ^{n}}{\sqrt{n!}}\left\vert 0\right\rangle _{A}\left\vert
0\right\rangle _{E}\\
&  =\frac{1}{\sqrt{n!}}\sum_{k=0}^{n}\binom{n}{k}\sqrt{\eta}^{k}\sqrt{1-\eta
}^{n-k}\left(  \hat{a}^{\dag}\right)  ^{k}\left(  \hat{e}^{\dag}\right)
^{n-k}\left\vert 0\right\rangle _{A}\left\vert 0\right\rangle _{E}\\
&  =\sum_{k=0}^{n}\sqrt{\binom{n}{k}}\sqrt{\eta}^{k}\sqrt{1-\eta}%
^{n-k}\left\vert k\right\rangle _{A}\left\vert n-k\right\rangle _{E},
\end{align*}
so that%
\begin{multline*}
U^{\otimes n}\left(  \left\vert a^{n}\right\rangle \left\langle b^{n}%
\right\vert \otimes\left\vert 0\right\rangle \left\langle 0\right\vert
^{\otimes n}\right)  \left(  U^{\dag}\right)  ^{\otimes n}\\
=\sum_{k_{1}=0}^{a_{1}}\cdots\sum_{k_{n}=0}^{a_{n}}\sum_{l_{1}=0}^{b_{1}%
}\cdots\sum_{l_{n}=0}^{b_{n}}\sqrt{\binom{a_{1}}{k_{1}}\cdots\binom{a_{n}%
}{k_{n}}\binom{b_{1}}{l_{1}}\cdots\binom{b_{n}}{l_{n}}}\times\\
\sqrt{\eta}^{\sum_{i=1}^{n}k_{i}+l_{i}}\sqrt{1-\eta}^{\sum_{i=1}^{n}%
a_{i}-k_{i}+b_{i}-l_{i}}\left\vert k_{1}\right\rangle \left\langle
l_{1}\right\vert _{A_{1}}\otimes\cdots\otimes\left\vert k_{n}\right\rangle
\left\langle l_{n}\right\vert _{A_{n}}\otimes\\
\left\vert a_{1}-k_{1}\right\rangle \left\langle b_{1}-l_{1}\right\vert
_{E_{1}}\otimes\cdots\otimes\left\vert a_{n}-k_{n}\right\rangle \left\langle
b_{n}-l_{n}\right\vert _{E_{n}}%
\end{multline*}
To evaluate the trace in (\ref{eq:term-to-sub-into}), we can perform it with
respect to the number basis, calculating
\[
\sum_{\substack{a^{n},b^{n}:\\\sum_{i}a_{i},\sum_{i}b_{i}\leq\left\lceil
nN_{S}\right\rceil }}\rho_{a^{n},b^{n}}\sum_{\substack{c^{n},d^{n}:\\\sum
_{i}c_{i}\leq\left\lceil n\left(  \eta N_{S}+\delta_{2}\right)  \right\rceil
}}\left\langle c^{n}\right\vert _{B^{n}}\left\langle d^{n}\right\vert _{E^{n}%
}U^{\otimes n}\left(  \left\vert a^{n}\right\rangle \left\langle
b^{n}\right\vert _{A^{n}}\otimes\left(  \left\vert 0\right\rangle \left\langle
0\right\vert \right)  _{E^{n}}^{\otimes n}\right)  \left(  U^{\dag}\right)
^{\otimes n}\left\vert c^{n}\right\rangle _{B^{n}}\left\vert d^{n}%
\right\rangle _{E^{n}}%
\]
as%
\begin{multline*}
\sum_{\substack{a^{n},b^{n}:\\\sum_{i}a_{i},\sum_{i}b_{i}\leq\left\lceil
nN_{S}\right\rceil }}\rho_{a^{n},b^{n}}\sum_{\substack{c^{n},d^{n}:\\\sum
_{i}c_{i}\leq\left\lceil n\left(  \eta N_{S}+\delta_{2}\right)  \right\rceil
}}\sum_{k_{1}=0}^{a_{1}}\cdots\sum_{k_{n}=0}^{a_{n}}\sum_{l_{1}=0}^{b_{1}%
}\cdots\sum_{l_{n}=0}^{b_{n}}\sqrt{\binom{a_{1}}{k_{1}}\cdots\binom{a_{n}%
}{k_{n}}\binom{b_{1}}{l_{1}}\cdots\binom{b_{n}}{l_{n}}}\times\\
\sqrt{\eta}^{\sum_{i=1}^{n}k_{i}+l_{i}}\sqrt{1-\eta}^{\sum_{i=1}^{n}%
a_{i}-k_{i}+b_{i}-l_{i}}\left\langle c_{1}|k_{1}\right\rangle \left\langle
l_{1}|c_{1}\right\rangle _{A_{1}}\times\cdots\times\left\langle c_{n}%
|k_{n}\right\rangle \left\langle l_{n}|c_{n}\right\rangle _{A_{n}}\times\\
\left\langle d_{1}|a_{1}-k_{1}\right\rangle \left\langle b_{1}-l_{1}%
|d_{1}\right\rangle _{E_{1}}\times\cdots\times\left\langle d_{n}|a_{n}%
-k_{n}\right\rangle \left\langle b_{n}-l_{n}|d_{n}\right\rangle _{E_{n}}%
\end{multline*}
By inspection, the only terms that survive are those for which $c_{i}%
=k_{i}=l_{i}$ and $d_{i}=a_{i}-k_{i}=b_{i}-l_{i}$ (and hence $a_{i}=b_{i}$),
leaving the above equal to%
\begin{equation}
\sum_{\substack{a^{n}:\\\sum_{i}a_{i}\leq\left\lceil nN_{S}\right\rceil }%
}\rho_{a^{n},a^{n}}\sum_{k_{1}=0}^{a_{1}}\cdots\sum_{k_{n}=0}^{a_{n}%
}\mathcal{I}\left[  \sum_{i=1}^{n}k_{i}\leq\left\lceil n\left(  \eta
N_{S}+\delta_{2}\right)  \right\rceil \right]  \binom{a_{1}}{k_{1}}%
\cdots\binom{a_{n}}{k_{n}}\eta^{\sum_{i=1}^{n}a_{i}-k_{i}}\left(
1-\eta\right)  ^{\sum_{i=1}^{n}k_{i}}, \label{eq:final-term-to-bound}%
\end{equation}
where $\mathcal{I}\left[  \cdot\right]  $ is an indicator function. To find a
lower bound on this expression, first we should realize that%
\begin{equation}
\sum_{k_{1}=0}^{a_{1}}\cdots\sum_{k_{n}=0}^{a_{n}}\mathcal{I}\left[
\sum_{i=1}^{n}k_{i}\leq\left\lceil n\left(  \eta N_{S}+\delta_{2}\right)
\right\rceil \right]  \binom{a_{1}}{k_{1}}\cdots\binom{a_{n}}{k_{n}}\eta
^{\sum_{i=1}^{n}a_{i}-k_{i}}\left(  1-\eta\right)  ^{\sum_{i=1}^{n}k_{i}}
\label{eq:final-final-term-to-bound}%
\end{equation}
is equal to one whenever $\sum_{i}a_{i}\leq\left\lceil n\left(  \eta
N_{S}+\delta_{2}\right)  \right\rceil $ because, in such a case, we are
guaranteed that $\sum_{i=1}^{n}k_{i}\leq\left\lceil n\left(  \eta N_{S}%
+\delta_{2}\right)  \right\rceil $. Thus, we should focus on the case in which%
\[
\left\lceil n\left(  \eta N_{S}+\delta_{2}\right)  \right\rceil <\sum_{i}%
a_{i}\leq\left\lceil nN_{S}\right\rceil .
\]
However, the expression in (\ref{eq:final-final-term-to-bound}) is related to
the probability that the average of a large number of independent Bernoulli
random variables is no larger than the mean of these random variables plus a
small offset. That is, after defining the i.i.d.~Bernoulli random variables
$X_{i,j}$ each with parameter $\eta$,\footnote{The nice physical
interpretation here is that each of the $\sum_{i}a_{i}$ i.i.d.~Bernoulli
random variables corresponds to a single photon that has a probability $\eta$
of making it through the beamsplitter.} such that the binomial random variable
$K_{i}=\sum_{j}X_{i,j}$, we find that the expression above is equal to%
\begin{equation}
\Pr\left\{  \sum_{i,j}X_{i,j}\leq\left\lceil n\left(  \eta N_{S}+\delta
_{2}\right)  \right\rceil \right\}  \geq\Pr\left\{  \sum_{i,j}X_{i,j}\leq
S\left(  \eta+\delta_{3}\right)  \right\}  , \label{eq:prob-concentration}%
\end{equation}
where $S\equiv\sum_{i}a_{i}$ is the total number of these Bernoulli random
variables (total number of photons)\ and the inequality in
(\ref{eq:prob-concentration}) follows from%
\begin{align*}
\left\lceil n\left(  \eta N_{S}+\delta_{2}\right)  \right\rceil  &  \geq
n\left(  \eta N_{S}+\delta_{2}\right) \\
&  =n\eta N_{S}+n\delta_{2}\\
&  \geq\eta\left(  \left\lceil nN_{S}\right\rceil -1\right)  +n\delta_{2}\\
&  \geq\left\lceil nN_{S}\right\rceil \left(  \eta+\delta_{3}\right) \\
&  \geq S\left(  \eta+\delta_{3}\right)  ,
\end{align*}
where for large enough $n$ we can choose a constant $\delta_{3}$ such that
$0<\delta_{3}\leq\frac{1}{\left\lceil nN_{S}\right\rceil }\left(  n\delta
_{2}-\eta\right)  $. We can then apply the Hoeffding concentration bound \cite{H63,T12} to
conclude that (\ref{eq:prob-concentration}) is larger than%
\[
1-\exp\left\{  -2\delta_{3}^{2}S\right\}  \geq1-\exp\left\{  -2\delta_{3}%
^{2}\eta nN_{S}\right\}  ,
\]
where the inequality follows from the assumption that $S\geq\left\lceil
n\left(  \eta N_{S}+\delta_{2}\right)  \right\rceil \geq n\eta N_{S}$. Then
this last expression converges to one exponentially fast with increasing $n$.
Thus, it follows that the expression in (\ref{eq:final-term-to-bound}) is
larger than%
\[
\left(  1-\delta_{1}\right)  \left(  1-\exp\left\{  -\left(  2\delta_{3}%
^{2}\eta N_{S}\right)  n\right\}  \right)  \geq1-\delta_{1}-\exp\left\{
-\left(  2\delta_{3}^{2}\eta N_{S}\right)  n\right\}  .
\]
Combining the above inequality with the one in (\ref{eq:first-block-sim-proof}%
) gives the statement of the lemma.
\end{proof}

\bigskip

With Lemmas~\ref{lem:energy-entropy-dimension} and \ref{lem:input-output-size}%
\ in hand, we can now easily prove the strong converse by an approach similar
to that in (\ref{eq:trivial-SC}). Indeed, let $\rho_{m}$ be the codewords of
any code for the pure-loss bosonic channel with transmissivity $\eta$ such
that the average codeword density operator has a large projection onto a space
with total photon number less than $\left\lceil nN_{S}\right\rceil $:%
\begin{equation}
\frac{1}{M}\sum_{m}\text{Tr}\left\{  \Pi_{\left\lceil nN_{S}\right\rceil }%
\rho_{m}\right\}  \geq1-\delta_{1}\left(  n\right)  ,
\label{eq:max-photon-number-constraint}%
\end{equation}
where $\delta_{1}\left(  n\right)  $ is a function decreasing to zero with
increasing $n$. Let $\left\{  \Lambda_{m}\right\}  $ denote a decoding
POVM\ acting on the output space of $n$ instances of the pure-loss bosonic channel.

\begin{theorem}
[Strong converse]The average success probability of any code satisfying
(\ref{eq:max-photon-number-constraint}) is bounded as follows:%
\[
\frac{1}{M}\sum_{m}\operatorname{Tr}\left\{  \Lambda_{m}\mathcal{N}^{\otimes
n}\left(  \rho_{m}\right)  \right\}  \leq2^{-n\left(  R-g\left(  \eta
N_{S}\right)  -\delta_{2}-\delta\right)  }+2\sqrt{\delta_{1}\left(  n\right)
+\exp\left\{  -\left(  2\delta_{3}^{2}\eta N_{S}\right)  n\right\}
+2\sqrt{\delta_{1}\left(  n\right)  }},
\]
where $\mathcal{N}^{\otimes n}$ denotes $n$ instances of the pure-loss bosonic
channel, $\delta$ is defined in Lemma~\ref{lem:energy-entropy-dimension} (with
$N_{S}$ replaced by $\eta N_{S}$), $\delta_{1}\left(  n\right)  $ is defined
in (\ref{eq:max-photon-number-constraint}), $\delta_{2}$ is an arbitrarily
small positive constant, and $\delta_{3}$ is defined in
Lemma~\ref{lem:input-output-size}. Thus, if $R>g\left(  \eta N_{S}\right)  $,
then we can pick $\delta_{2}$ and $\delta$ small enough such that $R>g\left(
\eta N_{S}\right)  +\delta_{2}+\delta$, and it follows that the success
probability of any family of codes satisfying (\ref{eq:max-photon-number-constraint})
decreases to zero in the limit of large $n$.
\end{theorem}

\begin{proof}
Consider that%
\begin{align*}
  \frac{1}{M}\sum_{m}\text{Tr}\left\{  \Lambda_{m}\mathcal{N}^{\otimes
n}\left(  \rho_{m}\right)  \right\} 
&  \leq\frac{1}{M}\sum_{m}\text{Tr}\left\{  \Lambda_{m}\Pi_{\left\lceil
n\left(  \eta N_{S}+\delta_{2}\right)  \right\rceil }\mathcal{N}^{\otimes
n}\left(  \rho_{m}\right)  \Pi_{\left\lceil n\left(  \eta N_{S}+\delta
_{2}\right)  \right\rceil }\right\} \\
&  \ \ \ \ \ \ \ \ \ \ +\frac{1}{M}\sum_{m}\left\Vert \mathcal{N}^{\otimes
n}\left(  \rho_{m}\right)  -\Pi_{\left\lceil n\left(  \eta N_{S}+\delta
_{2}\right)  \right\rceil }\mathcal{N}^{\otimes n}\left(  \rho_{m}\right)
\Pi_{\left\lceil n\left(  \eta N_{S}+\delta_{2}\right)  \right\rceil
}\right\Vert _{1}\\
&  \leq\frac{1}{M}\sum_{m}\text{Tr}\left\{  \Pi_{\left\lceil n\left(  \eta
N_{S}+\delta_{2}\right)  \right\rceil }\Lambda_{m}\Pi_{\left\lceil n\left(
\eta N_{S}+\delta_{2}\right)  \right\rceil }\mathcal{N}^{\otimes n}\left(
\rho_{m}\right)  \right\} \\
&  \ \ \ \ \ \ \ \ \ \ +2\sqrt{\delta_{1}\left(  n\right)  +\exp\left\{
-\left(  2\delta_{3}^{2}\eta N_{S}\right)  n\right\}  +2\sqrt{\delta
_{1}\left(  n\right)  }}.
\end{align*}
The first inequality is a consequence of (\ref{eq:trace-inequality}). The
second inequality follows from a variation of the Gentle Operator Lemma which
holds for ensembles~\cite{W99,ON07}. That is, for an ensemble $\left\{
p_{X}\left(  x\right)  ,\rho_{x}\right\}  $ for which $\sum_{x}p_{X}\left(
x\right)  $Tr$\left\{  \Lambda\rho_{x}\right\}  \geq1-\varepsilon$ for
$0\leq\varepsilon\leq1$, the following inequality holds $$\sum_{x}p_{X}\left(
x\right)  \Vert \rho_{x}-\sqrt{\Lambda}\rho_{x}\sqrt{\Lambda}\Vert
_{1}\leq2\sqrt{\varepsilon}.$$ We apply this along with
Lemma~\ref{lem:input-output-size} and the assumption in
(\ref{eq:max-photon-number-constraint}) to arrive at the second inequality.
Focusing on the first term in the last expression, we find the following upper bound by
a method similar to that in (\ref{eq:trivial-SC}):%
\begin{align*}
  \frac{1}{M}\sum_{m}\text{Tr}\left\{  \Pi_{\left\lceil n\left(  \eta
N_{S}+\delta_{2}\right)  \right\rceil }\Lambda_{m}\Pi_{\left\lceil n\left(
\eta N_{S}+\delta_{2}\right)  \right\rceil }\mathcal{N}^{\otimes n}\left(
\rho_{m}\right)  \right\} 
&  \leq\frac{1}{M}\sum_{m}\text{Tr}\left\{  \Pi_{\left\lceil n\left(  \eta
N_{S}+\delta_{2}\right)  \right\rceil }\Lambda_{m}\Pi_{\left\lceil n\left(
\eta N_{S}+\delta_{2}\right)  \right\rceil }\right\} \\
&  =M^{-1}\text{Tr}\left\{  \Pi_{\left\lceil n\left(  \eta N_{S}+\delta
_{2}\right)  \right\rceil }\right\} \\
&  \leq2^{-n\left(  R-g\left(  \eta N_{S}\right)  -\delta_{2}-\delta\right)
}.
\end{align*}
Here, the first inequality follows because $\left\Vert \mathcal{N}^{\otimes
n}\left(  \rho_{m}\right)  \right\Vert _{\infty}\leq1$. The following equality
holds because $\sum_{m}\Lambda_{m}=I$, and the last inequality follows from
Lemma~\ref{lem:energy-entropy-dimension}.
Putting everything
together, we arrive at the bound in the statement of the theorem.
\end{proof}

\section{Capacity-achieving codes meeting the maximum photon number constraint}

\label{sec:code-existence}The demand in (\ref{eq:max-photon-number-constraint}) 
seems like it is somewhat stringent. That is, do there actually exist
capacity-achieving codes
for the pure-loss bosonic channel that meet this constraint while being a
reliable scheme for communication?

This final section argues that there exist codes for the pure-loss bosonic
channel that meet the following two constraints:

\begin{enumerate}
\item The maximum photon number constraint in
(\ref{eq:max-photon-number-constraint}) is satisfied with $\delta_{1}\left(
n\right)  $ a function exponentially decreasing to zero in $n$.

\item The error probability is less than an arbitrarily small constant for
sufficiently large $n$.
\end{enumerate}
So our development here is a refinement of the coding theorem given in
Ref.~\cite{HolevoWerner2001,GGLMSY04}.

To be precise, for a given code, we would like for the following two 
conditions to hold
\begin{align}
\frac{1}{M}\sum_{m}\text{Tr}\left\{  \Pi_{\left\lceil nN_{S}\right\rceil }%
\rho_{m}\right\}   &  \geq1-\delta_{1}\left(  n\right)  ,\tag{$E_1$}\\
\frac{1}{M}\sum_{m}\text{Tr}\left\{  \Lambda_{m}\mathcal{N}^{\otimes n}\left(
\rho_{m}\right)  \right\}   &  \geq1-\varepsilon \tag{$E_2$},
\end{align}
for an arbitrarily small positive number $\varepsilon$  and sufficiently large $n$.
In order to prove the existence of codes satisfying these two constraints, we
pick codewords as $n$-fold tensor products of coherent states, randomly chosen
according to a complex, isotropic Gaussian distribution with variance
$N_{S}-\delta$ where $\delta$ is an arbitrarily small positive constant (the
same approach as given in Ref.~\cite{GGLMSY04}, following~\cite{HolevoWerner2001}).
We can then analyze the probability that the constraints $E_{1}$ and $E_{2}$ 
above are not met for a randomly chosen code:
\begin{align}
&\!\!\!\!\!\!\!\!\!\!\!\!\!\!\!\!
  \Pr\left\{  \left(  E_{1}\cap E_{2}\right)  ^{c}\right\} \nonumber\\
&  \leq\Pr\left\{  E_{1}^{c}\right\}  +\Pr\left\{  E_{2}^{c}\right\}
\nonumber\\
&  =\Pr\left\{  1-\frac{1}{M}\sum_{m}\text{Tr}\left\{  \Pi_{\left\lceil
nN_{S}\right\rceil }\rho_{m}\right\}  >\delta_{1}\left(  n\right)  \right\}
+\Pr\left\{  1-\frac{1}{M}\sum_{m}\text{Tr}\left\{  \Lambda_{m}\mathcal{N}%
^{\otimes n}\left(  \rho_{m}\right)  \right\}  >\varepsilon\right\}
\nonumber\\
&  \leq\frac{1}{\delta_{1}\left(  n\right)  }\mathbb{E}\left\{  1-\frac{1}%
{M}\sum_{m}\text{Tr}\left\{  \Pi_{\left\lceil nN_{S}\right\rceil }\rho
_{m}\right\}  \right\}  +\frac{1}{\varepsilon}\mathbb{E}\left\{  1-\frac{1}%
{M}\sum_{m}\text{Tr}\left\{  \Lambda_{m}\mathcal{N}^{\otimes n}\left(
\rho_{m}\right)  \right\}  \right\}.  \label{eq:union-markov}%
\end{align}
The first inequality is a consequence of the union bound and the second
follows from the Markov inequality. We know that for every constant
$\varepsilon^{2}$, the expectation of the average error probability of a
randomly chosen code is less than $\varepsilon^{2}$ as long as the rate is no
larger than $g\left(  \eta N_{S}\right)  $ and $n$ is sufficiently large
\cite{PhysRevA.54.1869,WGTS11}. This means that%
\begin{equation}
\mathbb{E}\left\{  1-\frac{1}{M}\sum_{m}\text{Tr}\left\{  \Lambda
_{m}\mathcal{N}^{\otimes n}\left(  \rho_{m}\right)  \right\}  \right\}
\leq\varepsilon^{2}. \label{eq:error-bound}%
\end{equation}
We can analyze the other term in (\ref{eq:union-markov}) as follows:%
\begin{equation}
\mathbb{E}\left\{  1-\frac{1}{M}\sum_{m}\text{Tr}\left\{  \Pi_{\left\lceil
nN_{S}\right\rceil }\rho_{m}\right\}  \right\}  =1-\frac{1}{M}\sum
_{m}\text{Tr}\left\{  \Pi_{\left\lceil nN_{S}\right\rceil }\mathbb{E}\left\{
\rho_{m}\right\}  \right\}  . \label{eq:meeting-photon-constraint}%
\end{equation}
The expected density operator $\mathbb{E}\left\{  \rho_{m}\right\}  $ is a
thermal state of mean photon number $N_{S}^{\prime}\equiv N_{S}-\delta$ since
we are choosing codewords as coherent states according to a complex Gaussian
distribution with variance $N_{S}^{\prime}$:%
\[
\mathbb{E}\left\{  \rho_{m}\right\}  =\theta\left(  N_{S}^{\prime}\right)
^{\otimes n},
\]
where%
\[
\theta\left(  N_{S}^{\prime}\right)  \equiv\int d^{2}\alpha\frac{1}{\pi
N_{S}^{\prime}}\exp\left\{  -\left\vert \alpha\right\vert ^{2}/N_{S}^{\prime
}\right\}  \ \left\vert \alpha\right\rangle \left\langle \alpha\right\vert
=\frac{1}{N_{S}^{\prime}+1}\sum_{l=0}^{\infty}\left(  \frac{N_{S}^{\prime}%
}{N_{S}^{\prime}+1}\right)  ^{l}\left\vert l\right\rangle \left\langle
l\right\vert .
\]
Observe that the distribution of a thermal state with respect to the photon
number basis is a geometric distribution with mean$~N_{S}^{\prime}$. Then
(\ref{eq:meeting-photon-constraint}) reduces to%
\[
1-\frac{1}{M}\sum_{m}\text{Tr}\left\{  \Pi_{\left\lceil nN_{S}\right\rceil
}\mathbb{E}\left\{  \rho_{m}\right\}  \right\}  =1-\text{Tr}\left\{
\Pi_{\left\lceil nN_{S}\right\rceil }\theta\left(  N_{S}^{\prime}\right)^{\otimes n}
\right\}  .
\]
Finally, the expression $1- \operatorname{Tr}\left\{  \Pi_{\left\lceil nN_{S}\right\rceil
}\theta\left(  N_{S}^{\prime}\right)^{\otimes n}  \right\}  $ is equal to the probability
that the average of a large number of independent geometric random variables deviate from
their mean by more than $\delta$, which is exponentially decreasing in
$n$ as%
\[
1-\text{Tr}\left\{  \Pi_{\left\lceil nN_{S}\right\rceil }\theta\left(
N_{S}^{\prime}\right)^{\otimes n}  \right\}  \leq [C\left(  \delta,N_S'\right)]^n    ,
\]
where $C\left(  \delta,N_S'\right)$ is a constant strictly less than one (see the
appendix for an explicit proof).
Thus, we can choose $\delta_{1}\left(  n\right)
=[C\left(  \delta,N_S'\right)]^{(n/2)}  $ (for example) and combine the above
bound with (\ref{eq:union-markov}) and (\ref{eq:error-bound}) to arrive at the
following upper bound%
\[
\Pr\left\{  \left(  E_{1}\cap E_{2}\right)^{c}\right\}  \leq
[C\left(  \delta,N_S'\right)]^{(n/2)}  +\varepsilon.
\]
So, for $n$ large enough and since $\delta$ can be an arbitrarily small
positive constant, this proves the existence of a code that satisfies the two
constraints given at the beginning of this section with a rate equal to
$g\left(  \eta N_{S}\right)  $. (In fact, the overwhelming fraction of codes
selected in this way satisfy these constraints for sufficiently large $n$,
while operating at the aforementioned rate.)

%

\section{Conclusion}

This paper has broadened the understanding of the classical
capacity of the pure-loss bosonic channel
by determining conditions under which a strong converse
theorem can and cannot hold. First, we
proved that there is a rate-error trade-off whenever there is only a mean photon
number constraint on the codewords input to the channels,
so that a strong converse theorem does not hold under such a constraint.
One can even use pure-state codewords to achieve this trade-off, which is an important distinction
between the classical and quantum theories of information for continuous variables.
Next, we proved that a strong converse theorem holds under a
particular maximum photon number constraint. Our proof was a simulation-based argument:
we first showed that it is possible to faithfully simulate the action of $n$ instances of
the pure-loss bosonic channel with transmissivity parameter $\eta \in [0,1]$, at a rate
of $g(\eta N_S)$ whenever it is guaranteed that the input to the channel has nearly all of its
shadow on a subspace having no more than $nN_S$ 
photons. By concatenating a channel code with the simulation code,
we showed that it would be possible to strongly violate the Holevo bound if one
could send classical data over the pure-loss bosonic channel at a rate $R > g(\eta N_S)$. 
Finally, we refined the coding theorem of~\cite{HolevoWerner2001} and~\cite{GGLMSY04},
to show that there exist coherent-state codes that achieve the classical capacity 
while satisfying our maximum photon number constraint.

\bigskip
{\bf Acknowledgements}. MMW is grateful to the quantum
information theory group at the Universitat Aut\`{o}noma de Barcelona for
hosting him for a research visit during April-May 2013, during which some of
the work for this research was completed. He is also grateful to the quantum information groups
at MIT and Raytheon BBN Technologies for hosting him as a
visitor during June-August 2013.
AW's work is supported
by the European Commission (STREP \textquotedblleft QCS\textquotedblright),
the European Research Council (Advanced Grant \textquotedblleft
IRQUAT\textquotedblright) and the Philip Leverhulme Trust;
furthermore by the Spanish MINECO,
project FIS2008-01236, with the support of FEDER funds.

\section*{Appendix}

Here we detail an explicit proof of the following bound:%
\[
\Pr\left\{  \frac{1}{n}\sum_{i=1}^{n}Z_{i}\geq\mu+\delta\right\}  \leq\left[
C\left(  \delta,p\right)  \right]  ^{n},
\]
where $C\left(  \delta,p\right)  $ is a constant strictly less than one,
$\delta>0$, and each $Z_{i}$ is an independent geometric random variable with
mean $\mu = \frac{p}{1-p}$ and probability mass function $\Pr\left\{
Z_{i}=k\right\}  =p^{k}\left(  1-p\right)  $ for $k\in\left\{  0,1,2,\ldots
\right\}  $. For this purpose, we use the well known \textquotedblleft
Bernstein trick\textquotedblright\ (exponential moment method)
\cite{T12}, according to which%
\[
\Pr\left\{  \frac{1}{n}\sum_{i=1}^{n}Z_{i}\geq\mu+\delta\right\}  \leq
\inf_{t>0}\left(  \frac{\mathbb{E}\left\{  \exp\left\{  tZ\right\}  \right\}
}{\exp\left\{  t\left(  \mu+\delta\right)  \right\}  }\right)  ^{n}.
\]
So we just need to find a $t$ for which $\mathbb{E}\left\{  \exp\left\{
tZ\right\}  \right\}  <\exp\left\{  t\left(  \mu+\delta\right)  \right\}  $.
The moment generating function $\mathbb{E}\left\{  \exp\left\{  tZ\right\}
\right\}  $ of a geometric random variable is%
\[
\mathbb{E}\left\{  \exp\left\{  tZ\right\}  \right\}  =\frac{1-p}{1-pe^{t}},
\]
where we require $t$ to be chosen so\ that $pe^{t}<1$ (i.e., $t<-\ln p$) in
order to ensure convergence of $\mathbb{E}\left\{  \exp\left\{  tZ\right\}
\right\}  $. Setting $x=e^{t}$, our problem is to find a value of $x$ for
which%
\[
\frac{1-p}{1-px}<x^{\left(  p/\left(  1-p\right)  +\delta\right)  }.
\]
At $x=1$, the left hand side\ is equal to the right hand side, and taking
derivatives, we find that%
\[
\left.  \frac{\partial}{\partial x}\left[  \frac{1-p}{1-px}\right]
\right\vert _{x=1}=\frac{p}{1-p},
\]
while%
\[
\left.  \frac{\partial}{\partial x}\left[  x^{\left(  p/\left(  1-p\right)
+\delta\right)  }\right]  \right\vert _{x=1}=\frac{p}{1-p}+\delta
\]
It holds that%
\[
\frac{p}{1-p} <\frac{p}{1-p}+\delta,
\]
so we can conclude that $x^{\left(  p/\left(  1-p\right)  +\delta\right)  }$
is growing faster than $\frac{1-p}{1-px}$ in a neighborhood of $1$, so that
there exists a value of $x<1/p$ (and thus $t$) such that%
\[
\frac{\mathbb{E}\left\{  \exp\left\{  tZ\right\}  \right\}  }{\exp\left\{
t\left(  \mu+\delta\right)  \right\}  }<1.
\]
We then set $C\left(  \delta,p\right)  =\mathbb{E}\left\{  \exp\left\{
tZ\right\}  \right\}  /\exp\left\{  t\left(  \mu+\delta\right)  \right\}  $
for this value of $t$.

\bibliographystyle{plain}
\bibliography{Ref}

\begin{thebibliography}{10}

\bibitem{BDHSW09}
Charles~H. Bennett, Igor Devetak, Aram~W. Harrow, Peter~W. Shor, and Andreas
  Winter.
\newblock Quantum reverse {Shannon} theorem.
\newblock 2009.
\newblock arXiv:0912.5537.

\bibitem{BBCW12}
Mario Berta, Fernando Brand{\~a}o, Matthias Christandl, and Stephanie Wehner.
\newblock Entanglement {C}ost of {Q}uantum {C}hannels.
\newblock In {\em Proceedings of the 2012 International Symposium on
  Information Theory}, pages 900--904, Cambridge, MA, USA, July 2012.

\bibitem{BBCW13}
Mario Berta, Fernando Brand{\~a}o, Matthias Christandl, and Stephanie Wehner.
\newblock Entanglement {C}ost of {Q}uantum {C}hannels.
\newblock {\em IEEE Transactions on Information Theory}, 59(10):6779--6795,
  October 2013.
\newblock arXiv:1108.5357.

\bibitem{CT06}
Thomas~M. Cover and Joy~A. Thomas.
\newblock {\em Elements of {I}nformation {T}heory}.
\newblock Wiley-Interscience, 2nd edition, 2006.

\bibitem{DHW05RI}
Igor Devetak, Aram~W. Harrow, and Andreas Winter.
\newblock A {R}esource {F}ramework for {Q}uantum {Shannon} {T}heory.
\newblock {\em IEEE Transactions on Information Theory}, 54(10):4587--4618,
  October 2008.
\newblock arXiv:quant-ph/0512015.

\bibitem{GGLMSY04}
Vittorio Giovannetti, Saikat Guha, Seth Lloyd, Lorenzo Maccone, Jeffrey~H.
  Shapiro, and Horace~P. Yuen.
\newblock Classical {C}apacity of the {L}ossy {B}osonic {C}hannel: {T}he
  {E}xact {S}olution.
\newblock {\em Physical Review Letters}, 92(2):027902, January 2004.
\newblock arXiv:quant-ph/0308012.

\bibitem{PhysRevA.54.1869}
Paul Hausladen, Richard Jozsa, Benjamin Schumacher, Michael Westmoreland, and
  William~K. Wootters.
\newblock Classical information capacity of a quantum channel.
\newblock {\em Physical Review A}, 54(3):1869--1876, September 1996.

\bibitem{H63}
Wassily Hoeffding.
\newblock Probability inequalities for sums of bounded random variables.
\newblock {\em Journal of the American Statistical Association},
  58(301):13--30, March 1963.

\bibitem{Hol98}
Alexander~S. Holevo.
\newblock The {C}apacity of the {Q}uantum {C}hannel with {G}eneral {S}ignal
  {S}tates.
\newblock {\em IEEE Transactions on Information Theory}, 44:269--273, 1998.

\bibitem{HolevoWerner2001}
Alexander~S. Holevo and Reinhard~F. Werner.
\newblock Evaluating capacities of {B}osonic {Gaussian} channels.
\newblock {\em Physical Review A}, 63:032312, February 2001.

\bibitem{KW09}
Robert Koenig and Stephanie Wehner.
\newblock A {S}trong {C}onverse for {C}lassical {C}hannel {C}oding {U}sing
  {E}ntangled {I}nputs.
\newblock {\em Physical Review Letters}, 103:070504, August 2009.
\newblock arXiv:0903.2838.

\bibitem{KWW12}
Robert Koenig, Stephanie Wehner, and J\"urg Wullschleger.
\newblock Unconditional {S}ecurity {F}rom {N}oisy {Q}uantum {S}torage.
\newblock {\em IEEE Transactions on Information Theory}, 58(3):1962--1984,
  March 2012.
\newblock arXiv:0906.1030.

\bibitem{N99}
Ashwin Nayak.
\newblock Optimal lower bounds for quantum automata and random access codes.
\newblock In {\em Proceedings of the 40th Annual Symposium on Foundations of
  Computer Science}, pages 369--376, New York City, NY, USA, October 1999.
\newblock arXiv:quant-ph/9904093.

\bibitem{ON99}
Tomohiro Ogawa and Hiroshi Nagaoka.
\newblock Strong {C}onverse to the {Q}uantum {C}hannel {C}oding {T}heorem.
\newblock {\em IEEE Transactions on Information Theory}, 45(7):2486--2489,
  November 1999.
\newblock arXiv:quant-ph/9808063.

\bibitem{ON07}
Tomohiro Ogawa and Hiroshi Nagaoka.
\newblock Making {G}ood {C}odes for {C}lassical-{Q}uantum {C}hannel {C}oding
  via {Q}uantum {H}ypothesis {T}esting.
\newblock {\em IEEE Transactions on Information Theory}, 53(6):2261--2266, June
  2007.

\bibitem{P10}
Yury Polyanskiy.
\newblock {\em Channel coding: {N}on-asymptotic fundamental limits}.
\newblock PhD thesis, Princeton University, November 2010.

\bibitem{PhysRevA.56.131}
Benjamin Schumacher and Michael~D. Westmoreland.
\newblock Sending classical information via noisy quantum channels.
\newblock {\em Physical Review A}, 56(1):131--138, July 1997.

\bibitem{S09}
Jeffrey~H. Shapiro.
\newblock {The quantum theory of optical communications}.
\newblock {\em IEEE Journal of Selected Topics in Quantum Electronics},
  15(6):1547--1569, 2009.

\bibitem{T12}
Terence Tao.
\newblock {\em Topics in Random Matrix Theory}, volume 132 of {\em Graduate
  Studies in Mathematics}.
\newblock American Mathematical Society, 2012.
\newblock See also
  \verb+http://terrytao.wordpress.com/2010/01/03/254a-notes-1-concentration-of-measure/+.

\bibitem{WPGCRSL12}
Christian Weedbrook, Stefano Pirandola, Ra\'{u}l Garc\'{\i}a-Patr\'{o}n,
  Nicolas~J. Cerf, Timothy~C. Ralph, Jeffrey~H. Shapiro, and Seth Lloyd.
\newblock Gaussian quantum information.
\newblock {\em Reviews of Modern Physics}, 84:621--669, May 2012.
\newblock arXiv:1110.3234.

\bibitem{WGTS11}
Mark~M. Wilde, Saikat Guha, Si-Hui Tan, and Seth Lloyd.
\newblock Explicit capacity-achieving receivers for optical communication and
  quantum reading.
\newblock In {\em Proceedings of the 2012 International Symposium on
  Information Theory}, pages 551--555, Boston, Massachusetts, USA, July 2012.
\newblock arXiv:1202.0518.

\bibitem{WWY13}
Mark~M. Wilde, Andreas Winter, and Dong Yang.
\newblock Strong converse for the classical capacity of entanglement-breaking
  and {Hadamard} channels.
\newblock June 2013.
\newblock arXiv:1306.1586.

\bibitem{W99}
Andreas Winter.
\newblock Coding {T}heorem and {S}trong {C}onverse for {Q}uantum {C}hannels.
\newblock {\em IEEE Transactions on Information Theory}, 45(7):2481--2485,
  1999.

\bibitem{WinterPhD99}
Andreas Winter.
\newblock {\em {Coding Theorems of Quantum Information Theory}}.
\newblock PhD thesis, Universit\"{a}t Bielefeld, July 1999.
\newblock arXiv:quant-ph/9907077.

\bibitem{W02}
Andreas Winter.
\newblock Compression of sources of probability distributions and density
  operators.
\newblock 2002.
\newblock arXiv:quant-ph/0208131.

\bibitem{YO93}
Horace~P. Yuen and Masanao Ozawa.
\newblock Ultimate {I}nformation {C}arrying {L}imit of {Q}uantum {S}ystems.
\newblock {\em Physical Review Letters}, 70:363--366, January 1993.

\end{thebibliography}

\end{document}